\documentclass[11pt]{article}

\usepackage{amsfonts}
\usepackage{amsmath}
\usepackage{amssymb}
\usepackage{amsthm}
\usepackage{mathrsfs}
\usepackage{graphicx}
\theoremstyle{plain}
\newtheorem{theorem}{Theorem}
\newtheorem{proposition}[theorem]{Proposition}

\newtheorem{corollary}[theorem]{Corollary}

\theoremstyle{definition}
\newtheorem{example}[theorem]{Example}

\numberwithin{exercise}{section}
\numberwithin{equation}{section}
\numberwithin{theorem}{section}
\numberwithin{problem}{section}

\numberwithin{figure}{section}
\DeclareMathOperator{\Int}{int}

\DeclareMathOperator{\diag}{diag}

\newcommand{\bs}[1]{{\boldsymbol{#1}}}

\newcommand{\R}{\mathbf{R}}

\begin{document}

\title{{Adaptive Fitness Landscape for Replicator Systems:\\ To Maximize or not to Maximize}}

\author{Alexander S. Bratus$^{1,2,}$\footnote{e-mail: alexander.bratus@yandex.ru}, Artem S. Novozhilov$^{{3},}$\footnote{e-mail: artem.novozhilov@ndsu.edu}, Yuri S. Semenov$^{2,}$\footnote{e-mail: yuri\_semenoff@mail.ru} \\[3mm]
\textit{\normalsize $^\textrm{\emph{1}}$Faculty of Computational Mathematics and Cybernetics,}\\[-1mm]
\textit{\normalsize Lomonosov Moscow State University, Moscow 119992, Russia}\\[2mm]
\textit{\normalsize $^\textrm{\emph{2}}$Applied Mathematics–1, Moscow State University of Railway Engineering,}\\[-1mm]
\textit{\normalsize Moscow 127994, Russia}\\[2mm]
\textit{\normalsize $^\textrm{\emph{3}}$Department of Mathematics, North Dakota State University, Fargo, ND 58108, USA}}

\date{}

\maketitle

%Abstract goes here
\begin{abstract}
Sewall Wright's adaptive landscape metaphor penetrates a significant part of evolutionary thinking. Supplemented with Fisher's fundamental theorem of natural selection and Kimura's maximum principle, it provides a unifying and intuitive representation of the evolutionary process under the influence of natural selection as the hill climbing on the surface of mean population fitness. On the other hand, it is also well known that for many more or less realistic mathematical models this picture is a sever misrepresentation of what actually occurs. Therefore, we are faced with two questions. First, it is important to identify the cases in which adaptive landscape metaphor actually holds exactly in the models, that is, to identify the conditions under which system's dynamics  coincides with the process of searching for a (local) fitness maximum. Second, even if the mean fitness is not maximized in the process of evolution, it is still important to understand the structure of the mean fitness manifold and see the implications of this structure on the system's dynamics. Using as a basic model the classical replicator equation, in this note we attempt to answer these two questions and illustrate our results with simple well studied systems.


\end{abstract}

\section{Introductory remarks and mathematical problem statement}

The infamous Fisher's \textit{theorem of natural selection} asserts that ``the rate of increase in [mean] fitness of any organism at any time is equal to its genetic variance in fitness at that time''\cite{fisher1930genetical}. This theorem, for which Fisher himself did not provide any mathematical justification, can be indeed proved within the simplest possible framework of natural selection acting on a haploid population with $n$ alleles in the absence of mutations and recombination, assuming that we deal with an infinite population. Let $u_i$ denote the frequency of the $i$-th allele and let $m_i$ be the corresponding Malthusian fitness, then a very simple \textit{replicator equation}
\begin{equation}\label{eq0:1}
\dot u_i=u_i(m_i-\bar m(\bs u)),\quad i=1,\ldots,n,
\end{equation}
where
$$
\bar{m}(\bs u)=\sum_{i=1}^n m_iu_i,
$$
describes the change with time of the vector of frequencies $\bs u(t)=(u_1(t),\ldots,u_n(t))\in S_n$ that belongs for any time moment $t$ to the simplex $$S_n=\{\bs u\in\R^n\colon u_i\geq 0,\,\sum_{i=1}^n u_i=1\}.$$ Considering $\bar m$ as a function of $t$ and calculating the derivative of this function along the orbits of our system, we get
$$
\dot{\bar m}(t)=\sum_{i=1}^n m_i u_i^2(t)-\left(\sum_{i=1}^n m_i u_i(t)\right)^2=\textsf{Var}(t)\geq 0,
$$
which can be taken literally as a mathematical form of Fisher's theorem.

Moreover, system \eqref{eq0:1} can be also used to illustrate another important metaphor of theoretical population genetics, --- namely, Wright's \textit{adaptive landscape} \cite{wright1932roles}. In words, Wright envisioned a surface of the mean fitness and argued that the mean fitness maxima represent the so-called adaptive picks, to which natural selection leads evolving populations. For system \eqref{eq0:1} we can give an illustration of this process by considering how the mean population fitness, which, according to the calculations above, is non-decreasing, changes with time. To be precise and introduce notations that we will use in the following, let $\gamma_t$ be an orbit of dynamical system \eqref{eq0:1} (see Fig. \ref{fig0:1}(a)), whose behavior is quite simple because, assuming there is unique maximum $m_1=\max\{m_1,\ldots,m_n\}$, we find $\bs u(t)\to \bs e_1$, where $\bs e_1=(1,0,\ldots,0)$. The mean fitness surface is given by
$$
\Sigma=\{(\bar m,\bs u)\colon \bar m(\bs u)=\sum_{i=1}^n m_i u_i,\,\bs u\in S_n\},
$$
and to each $\gamma_t\in S_n$ there corresponds a curve $\Gamma_t$ on $\Sigma$, which in this particular case will be approaching the highest vertex on the hyperplane $\bar m(\bs u)=\sum_{i=1}^n m_iu_i$ that also must satisfies the constraint $\bs u\in S_n$ (see Fig. \ref{fig0:1}).
\begin{figure}[!t]$\quad$
\includegraphics[width=0.5\textwidth]{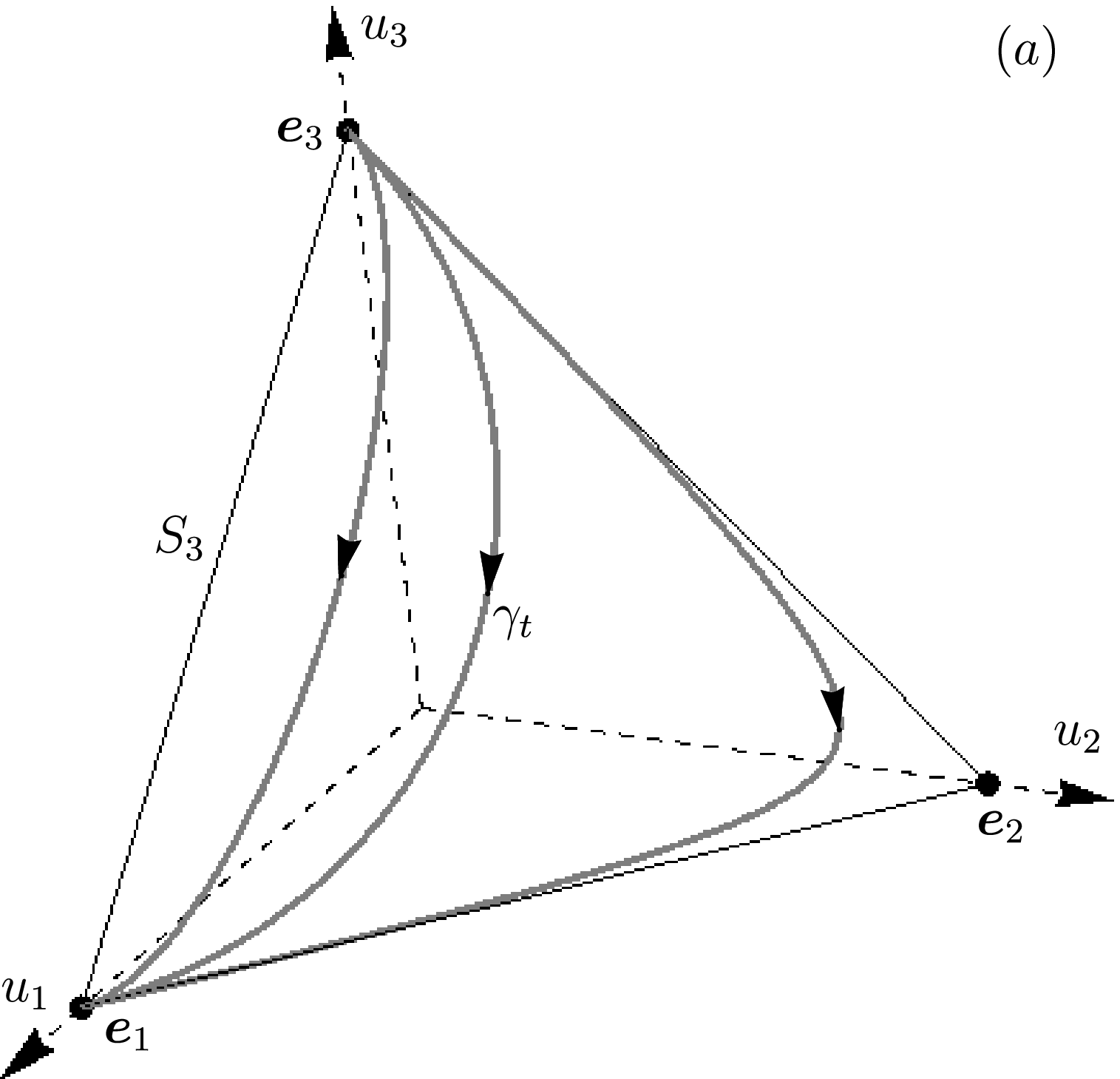}\hfill
\includegraphics[width=0.4\textwidth]{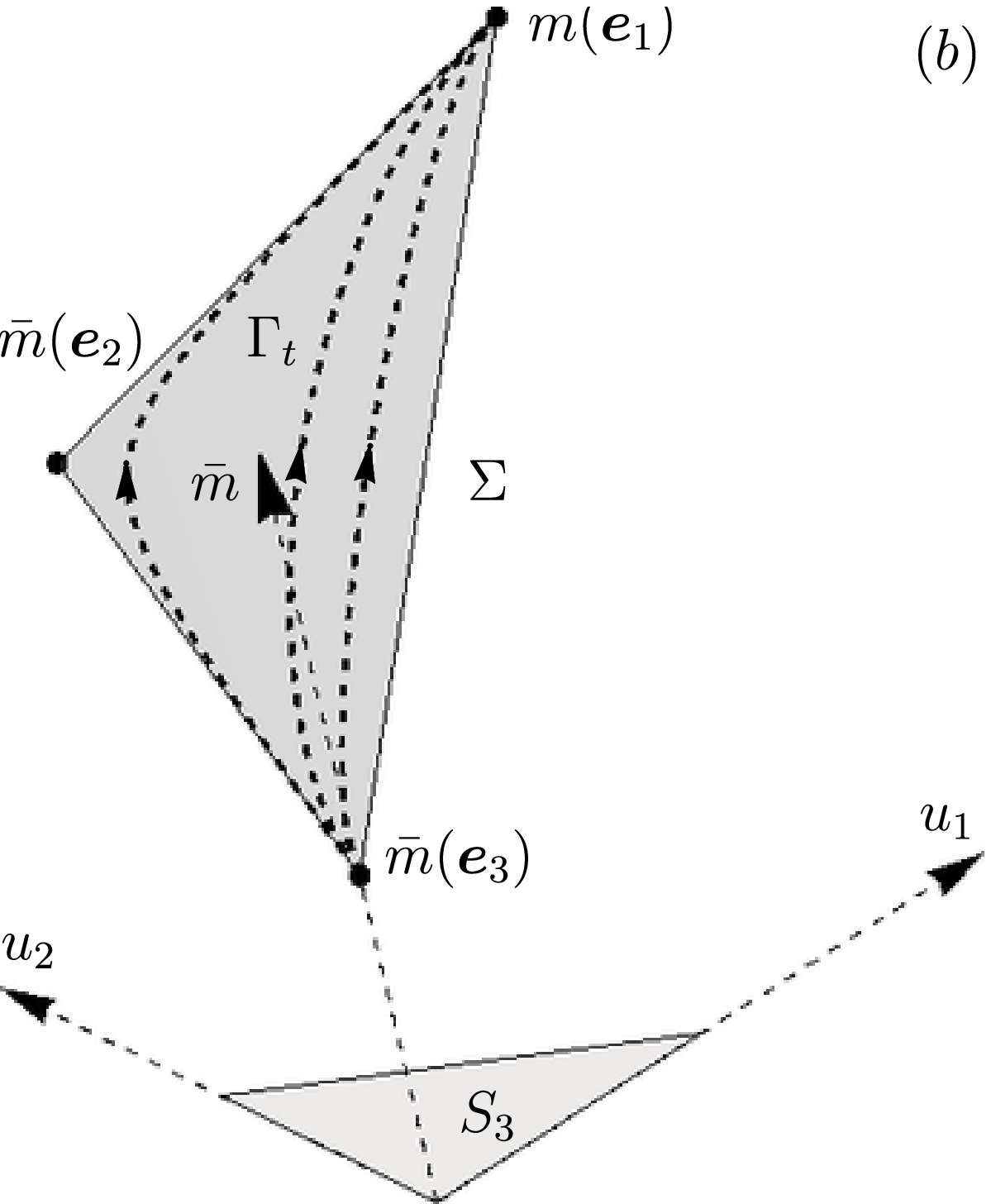}$\quad$
\caption{An illustration of the two important metaphors of evolutionary biology: Fisher's theorem of natural selection and Wright's adaptive fitness landscape. $(a)$ The phase space of system \eqref{eq0:1} in the case of $S_3$, the parameters are $m_1=3,m_2=2,m_1=1$. There are three equilibria, one of which, $\bs e_1$, is globally asymptotically stable. Three orbits are shown. $(b)$ The surface of the mean fitness $\Sigma$ that corresponds to the picture on the left. The three curves on $\Sigma$ are shown that correspond to the three orbits on the left. Note that the curves do not ``climb'' the mean fitness surface in the direction of the maximal fitness increase.}\label{fig0:1}
\end{figure}

Of course, Sewall Wright had a more general picture in mind. For example he admitted the existence of several adaptive picks, such that the population evolves in the direction of a nearest one, which may be only local and not global maximum. We can illustrate this idea using yet another extremely simple mathematical model, namely, the model of an autocatalytic growth. Now
\begin{equation}\label{eq0:2}
    \dot u_i=u_i(a_i u_i-\bar m(\bs u)),\quad i=1,\ldots,n,
\end{equation}
where now to guarantee that the simplex $S_n$ is invariant, we have that the mean population fitness is defined as
$$
\bar m(\bs u)=\sum_{i=1}^n a_i u_i^2.
$$
In this case, if we still stick to a haploid interpretation of our model, we now deal with a frequency dependent selection, albeit with its simplest case. Again, as in the previous example, it is straightforward to show that
$$
\dot{\bar {m}}(t)\geq 0,
$$
and all but one alleles disappear. Now, however, which allele will survive depends also on the initial conditions, as well as the parameters of the systems. Note that each vertex of the simplex is asymptotically stable and each has a basin of attraction of non-zero measure (see Fig. \ref{fig0:2}).
\begin{figure}[!t]$\quad$
\includegraphics[width=0.51\textwidth]{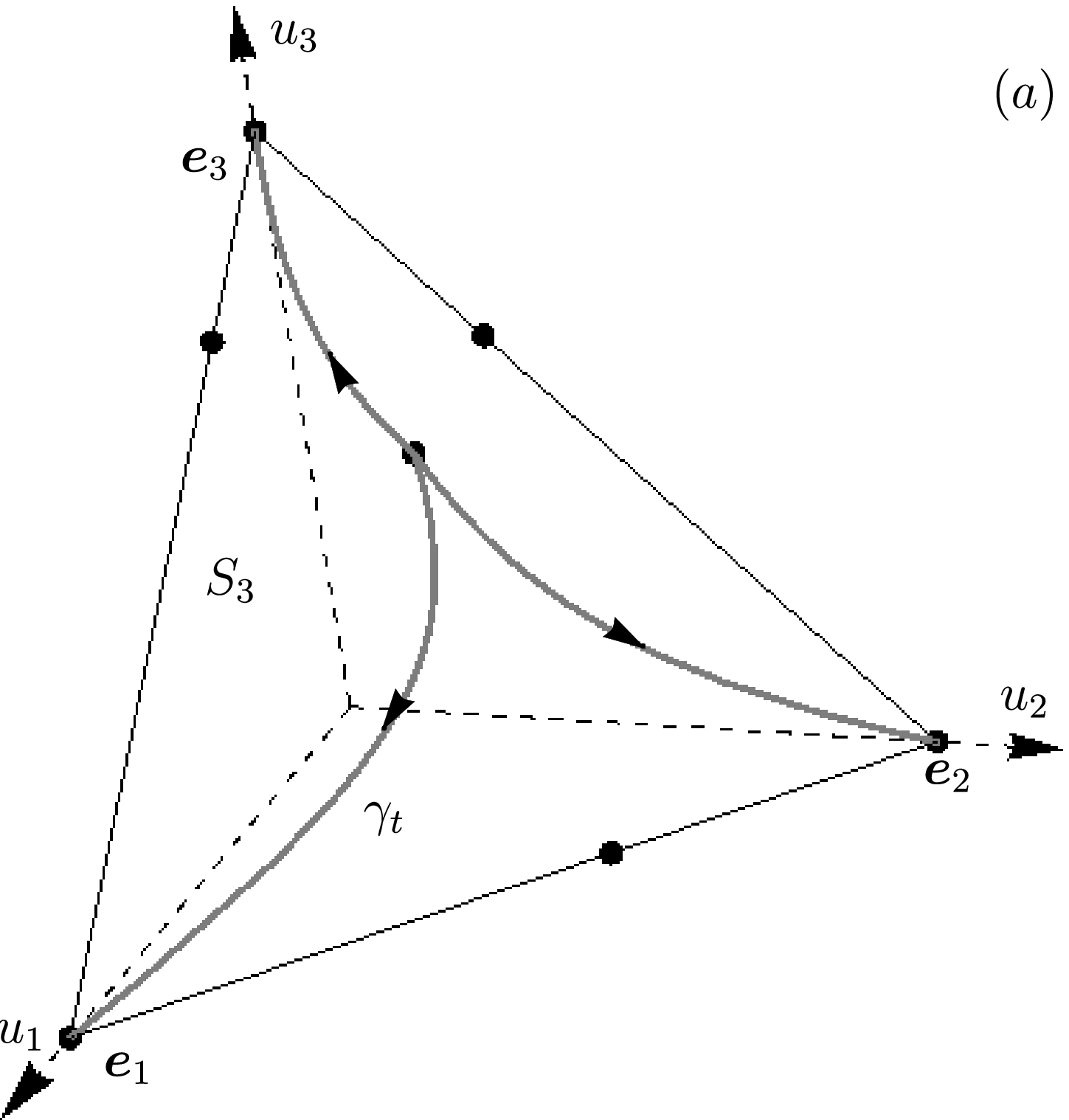}\hfill
\includegraphics[width=0.37\textwidth]{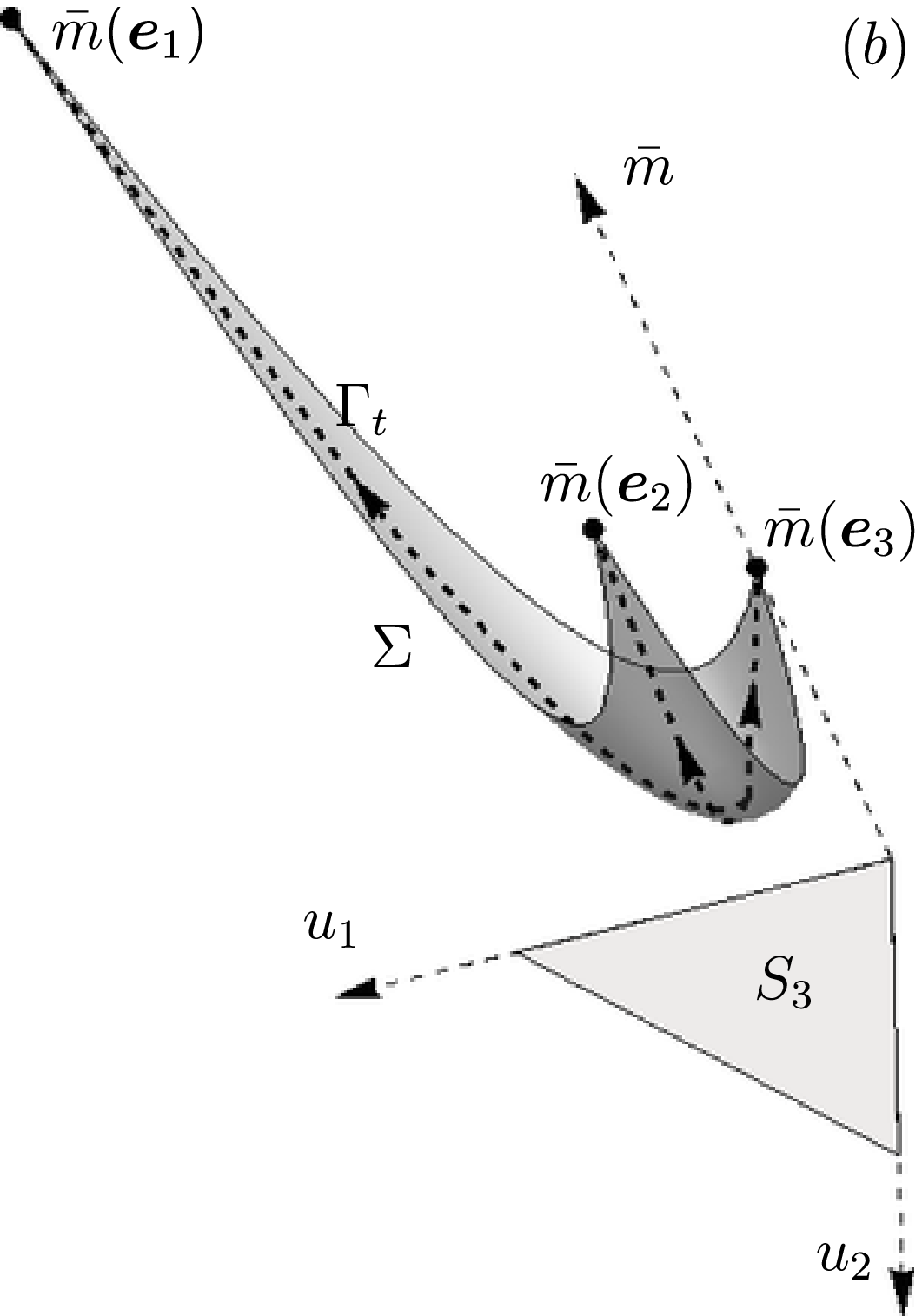}$\quad$
\caption{The phase space $(a)$ and adaptive fitness landscape $(b)$ of the replicator system with autocatalytic growth \eqref{eq0:2} on $S_3$. The dots denote the equilibria of the system. The parameters are $a_1=3,a_2=2,a_3=1$. There are three adaptive picks on the fitness landscape exactly at the vertices of the simplex, but the dynamics is still simple: Depending on the initial conditions, the population reaches the closest fitness maximum, at which all but one alleles are absent.}\label{fig0:2}
\end{figure}

Taken together, Wright's adaptive landscape and Fisher's theorem of natural selection can, in some sense, to be taken as a mathematical manifestation of a ``layman's'' understanding of how natural selection works. In short, the natural selection can be regarded as a process of fitness maximization (much more biological context can be found in \cite{birch2015natural}). On the other hand, it is well known in population genetics that the fitness is maximized only in very simple models, such as \eqref{eq0:1} or \eqref{eq0:2}, see an extensive discussion in \cite{birch2015natural} and references therein. It is not our goal in this short note to discuss various implications of these two seemingly mutually inconsistent points of view at the evolutionary process. We agree, however, with the general premise of \cite{birch2015natural} that it is important to reconcile them and, in particular, find maximization principles that apply conditionally. In other words, it is important to understand under which conditions our theoretical mathematical models imply fitness maximization.

For example, one such well known condition for the replicator equation of the form
$$
\dot u_i=u_i\Bigl(\sum_{j=1}^n a_{ij}u_j-\sum_{k,l=1}^n a_{kl}u_ku_l\Bigr)
$$
is that the matrix $\bs A=[a_{ij}]$ that defines this equation is symmetric, $\bs A=\bs A^\top$. In the language of the game theory we deal in this case with the so-called partnership games. It is known that in this case the mean fitness $\bar m(\bs u)=\sum_{i,j=1}^na_{ij} u_iu_j$ is 1) non-decreasing, since it is a strict Lyapunov function, and, even more, 2) it is a potential for the replicator equation, which hence becomes in this case a gradient system. In the language of theoretical population genetics it is said that in this case the state of the gene pool evolves under the influence of natural selection that not only the fitness increases, but also at a maximal rate (this is often called \textit{Kimura's maximum principle} \cite{kimura1958change}).

The nontrivial point here to see is that Kimura's maximum principle becomes a precise theorem only if one replaces the usual Euclidian metric, generated by the standard dot product, with another Riemannian metric generated by the so-called Shahshahani inner product \cite{shahshahani1979new,hofbauer1998ega}. This maximum principle can be extended for some particular cases of frequency depended selection \cite{sigmund1987maximum}, but in general we know that if $\bs A\neq \bs A^\top$ for the replicator equation, the mean fitness is not always maximized (actually, it can be shown that in this case the so-called KL-divergence between the initial and current distributions is minimized, see \cite{karev2010replicator} for information--theoretic approach to the replicator equation), and are left with an open question what are the necessary or/and sufficient conditions on the matrix $\bs A$ of the replicator equation, under which the dynamics of its solutions coincides with mean fitness maximization. In this note we incept such a study, derive several such conditions and illustrate them with examples of well studied systems.

\section{Replicator Equation}
Consider the now classical replicator equation (see, e.g., \cite{hofbauer1998ega,hofbauer2003egd} and references therein)
\begin{equation}\label{eq2:1}
    \dot u_i=u_i\bigl((\bs{Au})_i-\bs u\cdot\bs{Au} \bigr),\quad i=1,\ldots,n,
\end{equation}
where $\bs u\in S_n$, matrix $\bs A=[a_{ij}]$ is given and the mean fitness is
\begin{equation}\label{eq2:2}
\bar m(\bs u)=\bs u\cdot \bs{Au}=\sum_{i,j=1}^na_{ij}u_iu_j.
\end{equation}
Here and below the dot denotes the standard dot product in $\R^n$: $\bs x\cdot\bs y=\sum_{i=1}^n x_iy_i$. Together with solutions to \eqref{eq2:1} we also consider the surface of the mean fitness:
$$
\Sigma=\Bigl\{(\bar m,\bs u)\colon \bar{m}(\bs u)=\sum_{i,j=1}^n a_{ij}u_iu_j,\,\bs u\in S_n\Bigr\}.
$$
The orbits $\gamma_t$ of \eqref{eq2:1} correspond to the curves $\Gamma_t$ with directions on $\Sigma$, see the introductory section for simple examples (Figs. \ref{fig0:1} and  \ref{fig0:2}).

We represent matrix $\bs A$ as a sum of two matrices:
$$
\bs A=\frac 12(\bs A+\bs A^\top)+\frac 12(\bs A-\bs A^\top)=\bs B+\bs C,
$$
that is, as a sum of symmetric ($\bs B$) and skew-symmetric ($\bs C$) matrices. From \eqref{eq2:2} it follows that the mean fitness depends only on the symmetric part of $\bs A$, which is a simple yet key observation in our analysis:
\begin{equation}\label{eq2:3}
\bar{m}(\bs u)=\bs{Au}\cdot \bs u=\bs{Bu}\cdot \bs u.
\end{equation}

This means that there exists an orthogonal transformation $\bs U$ such that
$$
\bs U^\top\bs B\bs U=\bs \Lambda=\diag(\lambda_1,\ldots,\lambda_n),
$$
where $\lambda_i$ are real eigenvalues of $\bs B$.

This orthogonal transformation $\bs u=\bs U\bs w$ brings the quadratic form \eqref{eq2:3} to the canonical form
$$
\bar m(\bs w)=\sum_{i=1}^k \lambda_i^+w_i^2+\sum_{j=k+1}^n\lambda_j^-w_j^2,
$$
where we denote $\lambda_i^+$ and $\lambda_j^-$ the positive and negative eigenvalues of $\bs B$ respectively. The same transformation maps the simplex $S_n$ into the convex set
$$
W_n=\{\bs w\in\R^n\colon \bs w\cdot\bs {U}^\top \bs 1=1,\,\bs{Uw}\geq 0,\,\bs 1=(1,\ldots,1)\in \R^n\}.
$$
This means that the maximum value of the fitness surface $\Sigma$ can be obtained as a result of solving the problem of convex programming
$$
\sum_{i=1}^k\lambda_i^+w_i^2\longrightarrow \max,\quad \bs w\in W_n.
$$
Solution to this problem always exists but can be non-unique. An analogous statement holds for the minimal value on the surface $\Sigma$.

Yes another approach to find the extrema on $\Sigma$ is to use the classical method of Lagrange multipliers. We write down Lagrange's function
$$
\mathcal L(\bs u)=\bs{Bu}\cdot \bs u-\mu(\bs u\cdot\bs 1-1),\quad \mu\in\R.
$$
The necessary condition for a local extremum takes the form
\begin{equation}\label{eq2:4}
    \bs{Bu}=\mu\bs 1.
\end{equation}
The equality $\bs u\cdot \bs 1=1$ implies that $\mu=\bs{B^{-1}\bs 1\cdot\bs 1}$. Vector $\bs u\in S_n$ is a local maximum (minimum) if the condition of being negative definite (respectively, positive definite) for the quadratic form
\begin{equation}\label{eq2:5}
    \bs{Bw}\cdot\bs w<0\quad (>0)
\end{equation}
is satisfied for any $\bs w\in\R^n $ for which
\begin{equation}\label{eq2:6}
    \bs w\cdot\bs 1=0
\end{equation}
holds.

Consider a basis in the orthogonal complement to $\bs 1$:
$$
\bs r_1=(1,-1,0,\ldots,0),\quad \bs r_2=(1,0,-1,\ldots,0),\quad\ldots,\quad\bs r_{n-1}=(1,0,\ldots,-1).
$$
Then $\bs w=\sum_{k=1}^{n-1}\xi_k\bs r_k$, and

\begin{equation}\label{eq2:7}
\bs{Bw}\cdot \bs w=\sum_{k=1}^{n-1} \beta_{ij}\xi_i\xi_j,\quad \beta_{ij}=\bs{Br}_i\cdot\bs r_j.
\end{equation}
Therefore, checking the condition of sign defiteness of the form \eqref{eq2:5} boils down to checking the sign defitness of the form \eqref{eq2:7}, which is generated by the matrix with elements $\beta_{ij}=\bs{Br}_i\cdot\bs r_j$.

In particular, condition \eqref{eq2:5} implies the following proposition, which also connects the notion of an evolutionary stable state with the fact that an equilibrium of \eqref{eq2:1} coincides with a local strict maximum of the fitness surface.
\begin{proposition}If an equilibrium $\bs{\hat u}\in S_n$ of the replicator equation \eqref{eq2:1} is evolutionary stable then it is a strict local maximum of the fitness surface $\Sigma$.
\end{proposition}
\begin{proof}$\bs{\hat u}$ is evolutionary stable means (e.g., \cite{hofbauer1998ega})
$$
\bs{\hat u}\cdot\bs{Au}>\bs u\cdot \bs{Au}
$$
for all $\bs u$ in some neighborhood of $\bs{\hat u}$. Take $\bs u=\bs{\hat u}+\epsilon \bs w, \, \epsilon>0$. Since we require $\bs u\in S_n$ then $\bs w\cdot \bs 1=0$. The condition for the evolutionary stability implies
$$
\epsilon(\bs w\cdot\bs{A\hat u})+\epsilon^2(\bs w\cdot\bs{Aw})<0.
$$
Since $\bs{A\hat u}=\bar m(\bs{\hat u})\bs 1$, then the first term is zero and hence
$$
\bs w\cdot\bs{Aw}=\bs w\cdot\bs{Bw}<0,
$$
which concludes the proof.
\end{proof}
In other words, the existence of a strict maximum of the surface $\Sigma$ is a necessary condition for an equilibrium to be evolutionary stable.

Now let $\gamma_t$ be an orbit of \eqref{eq2:1} that approaches the equilibrium $\bs{\hat u}\in\Int S_n=\{\bs u\in S_n\colon u_i>0,\,i=1,\ldots,n\}$. Consider the conditions under which this would correspond to the case when the curve $\Gamma_t$ on the surface $\Sigma$ approaches an extremal point.
\begin{theorem}\label{theor2:2}An isolated equilibrium $\bs{\hat u}\in\Int S_n$ of \eqref{eq2:1} coincides with an extremal point on the surface $\Gamma$ if and only if the vectors that compose the rows of the skew symmetric matrix $\bs C=[c_{ij}]$ can be represented as linear combinations of the vectors that compose the rows of $\bs B=[b_{ij}]$:
\begin{equation}\label{eq2:8}
\bs c_i=\sum_{k=1}^n\mu_{ik}\bs b_k,\quad \bs b_k=(b_{k1},\ldots,b_{kn}),\quad \bs c_i=(c_{i1},\ldots,c_{in}),
\end{equation}
and, additionally,
\begin{equation}\label{eq2:9}
    \sum_{k=1}^n\mu_{ik}=0,\quad i=1,\ldots,n.
\end{equation}
\end{theorem}
\begin{proof}For the equilibrium point we have
\begin{equation}\label{eq2:10}
    \bs{A\hat u}=\bs{B\hat u}+\bs{C\hat u}=\bar m(\bs{\hat u})\bs 1.
\end{equation}
The solution to \eqref{eq2:4} is also a solution to \eqref{eq2:10} if and only if
\begin{equation}\label{eq2:11}
    \bs{C\hat u}=0,\quad \mu=\bar m(\bs{\hat{u}}).
\end{equation}
Indeed, the sufficiency of \eqref{eq2:11} is clear. Let us prove its necessity. We multiply the equalities \eqref{eq2:4} and $\eqref{eq2:10}$ by $\bs{\hat u}$ and use the fact $\bs{C\hat u}\cdot\bs{\hat u}=0$ to obtain $\bs{B\hat u}\cdot\bs{\hat u}=\bar m(\bs{\hat{u}})=\mu$ and therefore  \eqref{eq2:11} holds. Let us show now that in order for the solutions of
\begin{equation}\label{eq2:12}
    \bs{B\hat u}=\bar m(\bs{\hat{u}})\bs 1
\end{equation}
be in the kernel of $\bs C$, it is necessary and enough for the conditions \eqref{eq2:8} and \eqref{eq2:9} to hold. Assume that \eqref{eq2:8} and \eqref{eq2:9} hold. That is
$$
\bs C=\bs{\mathcal M}\bs B,\quad \bs{\mathcal M}=[\mu_{ij}].
$$
Hence
\begin{equation}\label{eq2:13}
    \bs{C\hat u}=\bs{\mathcal M B\hat u}=\bar m(\bs{\hat{u}})\bs{\mathcal M}\bs 1=0.
\end{equation}
In the opposite direction, since the rows of $\bs B$ linearly independent, we can always write \eqref{eq2:8} with some $\mu_{ij}$. Let $\bs{\hat u}$ solve \eqref{eq2:10}, therefore \eqref{eq2:13} holds, since $\bs{C\hat u}=0$, therefore, we must have $\bs{\mathcal M}\bs 1=0$, that is, the equalities \eqref{eq2:9}.
\end{proof}

\begin{corollary}Assume that conditions of Theorem \ref{theor2:2} hold and $\bs{\hat u}\in\Int S_n$ be a strict local maximum of the mean fitness. Then $\bs{\hat u}$ is an evolutionary stable state.
\end{corollary}
\begin{proof}We have for any $\bs u\in S_n$ that
$$
\bs{\hat u}\cdot\bs{Au}=\bs{\hat u}\cdot\bs{Bu}+\bs{\hat u}\cdot\bs{C u}=\bs{B\hat u}\cdot\bs u-\bs{C\hat u}\cdot\bs u=\bs{B\hat u}\cdot\bs u.
$$
On the other hand
$$
\bar m(\bs{\hat u})=\bs{B\hat u}\cdot \bs u=\bs{A\hat u}\cdot\bs{\hat u}.
$$
By assumption
$$
\bs{A\hat{u}}\cdot \bs{\hat u}>\bs u\cdot\bs{Au},
$$
therefore
$$
\bs{\hat u}\cdot\bs{Au}>\bs u\cdot \bs{A u}.
$$
\end{proof}

\begin{corollary}Assume that conditions of Theorem \ref{theor2:2} hold. Then the graph of the curve $\Gamma_t$ on the surface $\Sigma$ has neither local minima nor local maxima for $\bs u(t)\neq \bs{\hat u},\,\bs u(t)\in \Int S_n$.
\end{corollary}
\begin{proof}Taking the derivative of the mean fitness along the orbits of  \eqref{eq2:1}, we get
\begin{equation}\label{eq2:14}
    \dot{\bar{m}}(\bs u(t))=\sum_{i=1}^n(\bs b_i\cdot \bs u(t))u_i(t)-\bar m^2(\bs u(t))+\sum_{i=1}^n(\bs b_i\cdot\bs u(t))(\bs c_i\cdot\bs u(t))u_i(t).
\end{equation}
Looking for a contradiction, assume that $\bs u^\ast\in\Int S_n$ be a local extremum of $\bar m$ for some $t=t_\ast$ and therefore the condition \eqref{eq2:4} holds with $\bs u=\bs u^\ast$. Consider the vector
$$
\bs u_\lambda=\lambda \bs u^\ast+(1-\lambda)\bs{\hat u},\quad 0\leq \lambda\leq 1.
$$
Since the simplex is convex, $\bs u_\lambda$ belongs to $\Int S_n$ for any admissible $\lambda$. On the other hand, from \eqref{eq2:4} it follows
$$
\bs b_i\cdot \bs u^\ast=\mu,\quad i=1,\ldots,n,
$$
and conditions \eqref{eq2:8} and \eqref{eq2:9} imply
$$
\bs c_i\cdot\bs u^\ast=\sum_{k=1}^n\mu_{ik}(\bs b_k\cdot\bs u^\ast)=\mu\sum_{k=1}^n\mu_{ik}=0.
$$
Similarly, for $\bs{\hat u}$ we have
$$
\bs b_i\cdot\bs{\hat u}=\bar m,\quad\bs c_i\cdot\bs{\hat u}=0,\quad i=1,\ldots,n,
$$
which means that
$$
\dot{\bar m}(\lambda\bs u^\ast+(1-\lambda)\bs{\hat u})=0,
$$
and hence $\bar m$ is constant along the interval  corresponding to $0\leq \lambda\leq 1$, which implies that $\mu=\bar m$ and hence $\bs u^\ast=\bs{\hat u}$.
\end{proof}

\section{Examples}
In this section we consider two simple examples that illustrate the theoretical results we obtained in the previous section.
\begin{example}\label{ex3:1}
Consider a \textit{hypercycle equation} \cite{eigenshuster} in $S_3$ with the matrix
$$
\bs A=\begin{bmatrix}
        0 & 0 & k_1 \\
        k_2 & 0 & 0 \\
        0 & k_3 & 0 \\
      \end{bmatrix}.
$$
The coordinates of the internal equilibrium are given by
$$
\hat u_i=\frac{\bar m}{k_{i+1}}\,,\quad i=1,2,3,\,k_4=k_1,\quad \bar m=\left(\sum_{i=1}^3\frac{1}{k_i}\right)^{-1}\,.
$$
Simple calculations yield
$$
\bs B=\frac 12\begin{bmatrix}
                0 & k_2 & k_1 \\
                k_2 & 0 & k_3 \\
                k_1 & k_3 & 0 \\
              \end{bmatrix},\quad \bs C=\frac 12\begin{bmatrix}
                                                  0 & -k_2 & k_1 \\
                                                  k_2 & 0 & -k_3 \\
                                                  -k_1 & k_3 & 0 \\
                                                \end{bmatrix}.
$$
The extremal point of the fitness surface is given by
$$
u_1=\frac{k_3(k_1+k_2)-k_3^2}{\sigma}\,,\quad u_2=\frac{k_1(k_2+k_3)-k_1^2}{\sigma}\,,\quad u_3=\frac{k_2(k_1+k_3)-k_2}{\sigma}\,,
$$
where
$$
\sigma=2(k_1k_2+k_2k_3+k_1k_3)-(k_1^2+k_2^2+k_3^2).
$$
We can use equation \eqref{eq2:7} to check that this point is actually a local maximum. We obtain the quadratic form
$$
\bs{Bw}\cdot\bs w=-2k_2\xi_1^2+(k_1+k_2)\xi_1\xi_2+k_1\xi_2^2-k_3\xi_2\xi_3,
$$
which is negative definite if
$$
2k_3(k_1+k_2)>(k_1-k_2)^2.
$$
In particular, the last condition holds if, for example, $k_1=0.25,k_2=0.3,k_3=0.35$. In this case
$$
\bs{\hat u}=(0.3272,0.2803,0.3925),
$$
whereas the fitness surface maximum is at the point $\bs u_m=(0.2692,0.3846,0.3462)$, moreover
$$
\bar{m}(\bs{\hat u})<\bar m(\bs u_m),
$$
which implies that the curve on the surface $\Sigma$ can both increase and decrease while the system moving along trajectory $\gamma_t$. Let us show in this example that if the conditions \eqref{eq2:8}, \eqref{eq2:9} hold then $\bs u$ and $\bs{\hat u}$ coincide.

We have
$$
\bs c_1=\frac{k_1}{k_3}\bs b_1-\frac{k_2}{k_3}\bs b_3,\quad \bs c_2=-\frac{k_3}{k_1}\bs b_1+\frac{k_2}{k_1}\bs b_3,\quad \bs c_3=\frac{k_3}{k_2}\bs b_1-\frac{k_1}{k_2}\bs b_2.
$$
Conditions \eqref{eq2:8}, \eqref{eq2:9} satisfied if and only if $k_1=k_2=k_3=k$. In this case, putting $k=1$, the eigenvalues of $\bs B$ are $\lambda_1=1,\lambda_2=\lambda_3=1/2$. The orthogonal transformation $\bs{u}=\bs{Uw}$ that brings the mean fitness to the canonical form $w_1^2-\frac 12(w_2^2+w_3^2)$ takes the form
$$
\bs U=\begin{bmatrix}
        \sqrt{3}/3 & 0 & 2/\sqrt{6} \\
        \sqrt{3}/3 & -\sqrt{2}/2 & -1/\sqrt{6} \\
        \sqrt{3}/3 & \sqrt{2}/2 & -1/\sqrt{6} \\
      \end{bmatrix}.
$$
The condition $\bs u\cdot\bs 1=\bs w\cdot\bs U^\top\bs 1=1$ implies the constraint $w_1=\sqrt{3}/3$ and simplex $S_3$ is being transformed into the convex set
$$
W_3=\Bigl\{\bs w\colon\frac{\sqrt{2}}{2}w_2-\frac{1}{\sqrt{6}}w_3+\frac 13\geq 0,\,\frac{\sqrt{2}}{2}w_2-\frac{1}{\sqrt{6}}w_3+\frac 13\geq 0,\,w_3\geq -\frac{\sqrt{6}}{6}\Bigr\}.
$$
That it, in new coordinates, the fitness surface if given by the two-sheeted hyperboloid
$$
\bar m(\bs w)=\frac 13-\frac 12(w_2^2+w_3^2),
$$
which is a conus with the vertex at the point $w_2=w_3=0$. In this particular case the equilibrium point coincides with the maximum of the fitness surface. Any orbit of the system corresponds to the curve $\Gamma_t$ on the surface $\Sigma$, which climbs the conus for $t\to\infty$ (see also Fig. \ref{fig0:3} for a graphical illustration).
\begin{figure}[!t]$\quad$
\includegraphics[width=0.37\textwidth]{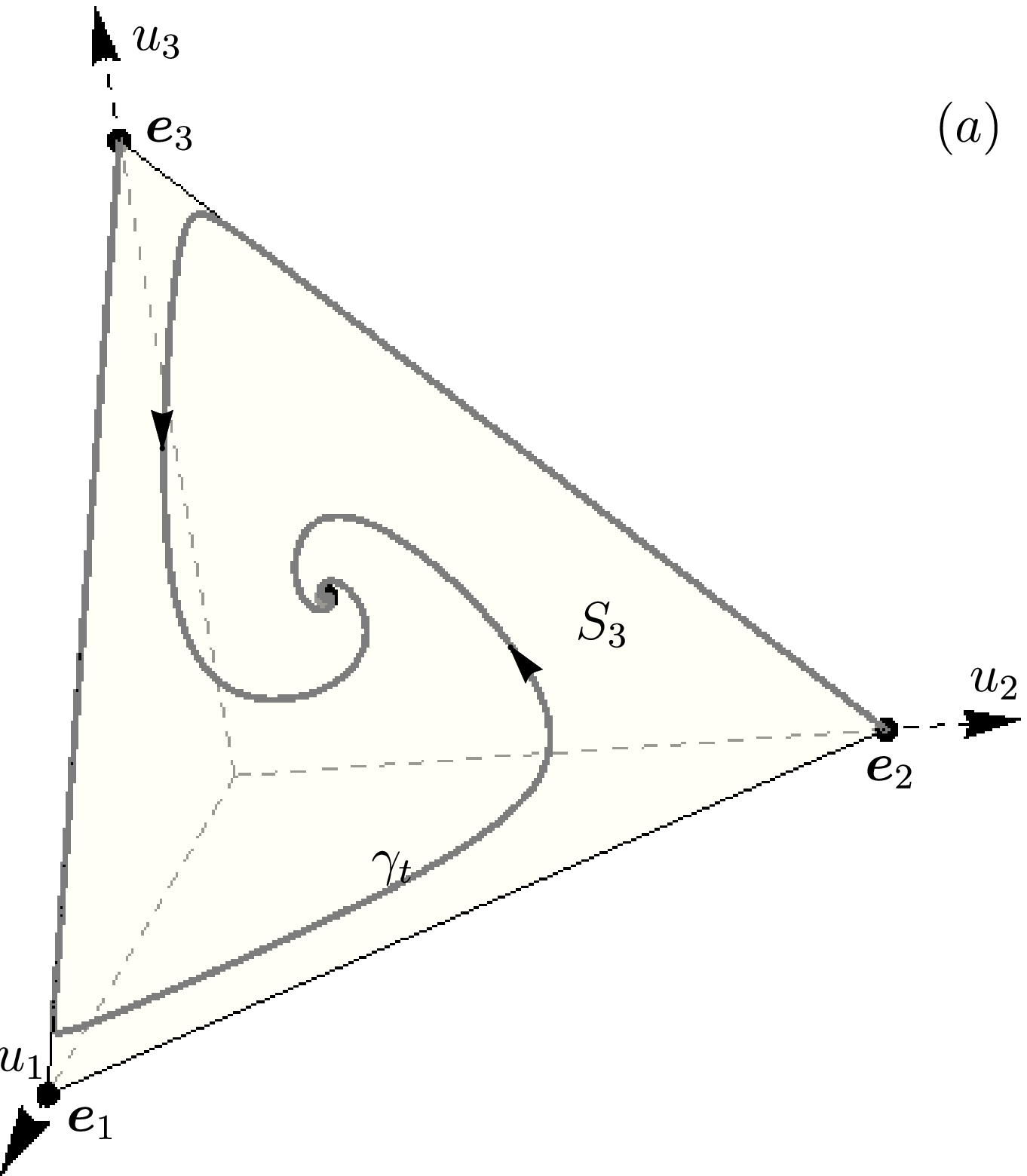}\hfill
\includegraphics[width=0.53\textwidth]{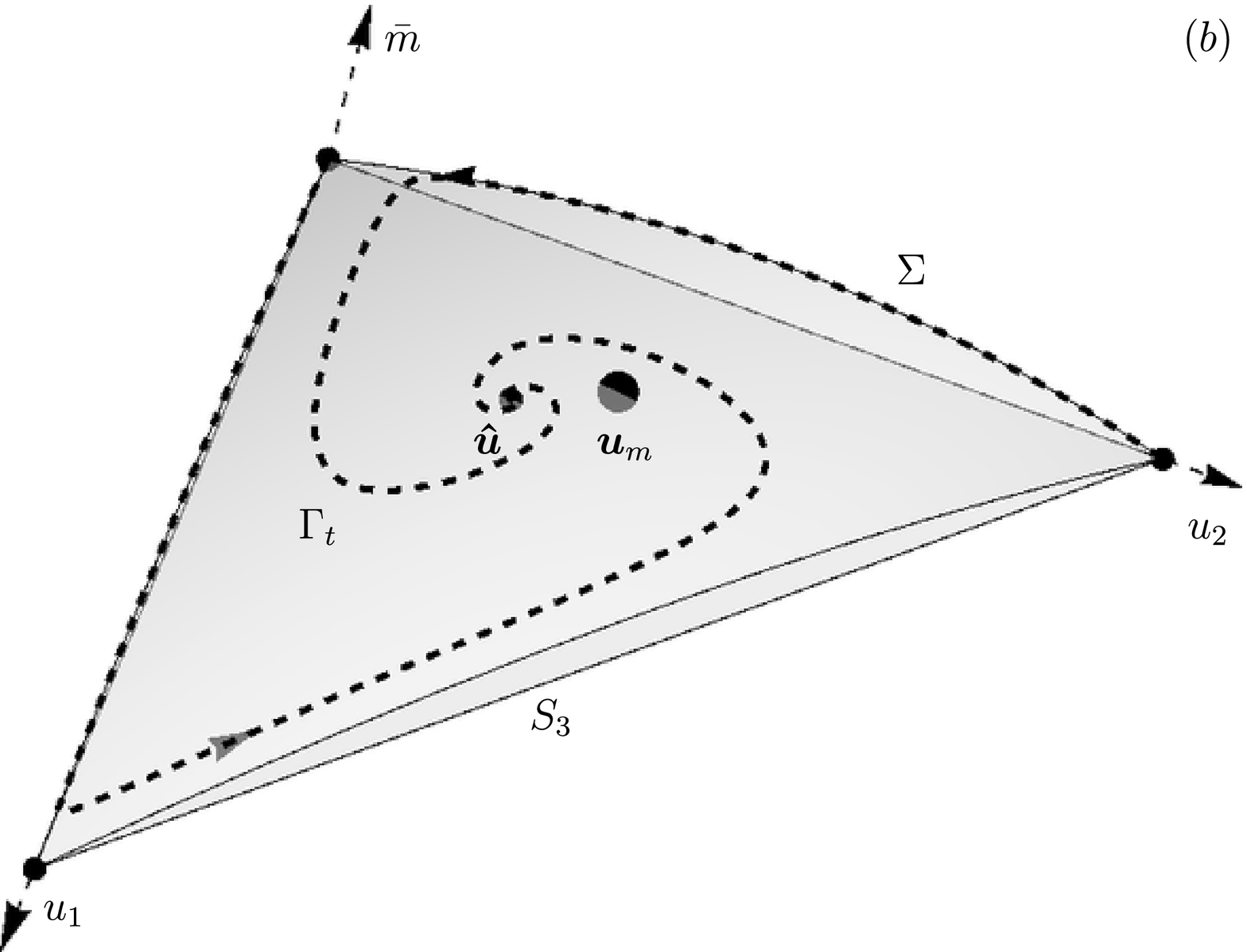}$\quad$
\caption{The phase space $(a)$ and adaptive fitness landscape $(b)$ of the hypercyclic system in Example \ref{ex3:1} on $S_3$. The parameters are $k_1=0.25,k_2=0.3,k_3=0.35$. Note that although there exists a strict Lyapunov function and the internal equilibrium $\bs{\hat u}$ is globally asymptotically stable, it is not evolutionary stable and does not coincide with the global maximum $\bs{u}_m$ on the mean fitness surface $\Sigma$.}\label{fig0:3}
\end{figure}

\end{example}

\begin{example}\label{ex3:2}
It is quite interesting to consider examples of fitness surfaces for which there are several local extrema and moreover periodic solutions exist.

Consider a replicator equation in $S_4$ with the matrix (we take this example from \cite{hofbauer1980dynamical})
$$
\bs A=\begin{bmatrix}
        0 & 0 & -\mu & 1 \\
        1 & 0 & 0 & -\mu \\
        -\mu & 1 & 0 & 0 \\
        0 & -\mu & 1 & 0 \\
      \end{bmatrix}.
$$
It is known that for $|\mu|<1$ the system is permanent, that is, its solutions are separated from the boundaries of the simplex, and for $\mu=0$ Poincar\'{e}--Andronov--Hopf bifurcation happens with appearance of a stable limit cycle for $\mu>0$. If $|\mu|>1$ then the system stops being permanent and its omega limit set is the heteroclinic cycle that contains the vertices of $S_4$.

In this case we have
$$
\bs B=\begin{bmatrix}
        0 & 1 & -2\mu & 1 \\
        1 & 0 & 1 & -2\mu \\
        -2\mu & 1 & 0 & 1 \\
        1 & -2\mu & 1 & 0 \\
      \end{bmatrix},
$$
and the eigenvalues of $\bs B$ are $\lambda_1=1-\mu,\lambda_{2,3}=\mu,\lambda_4=-(1+\mu)$. The orthogonal transformation is explicitly given by
$$
\bs U=\begin{bmatrix}
        1/2 & -\sqrt{2}/2 & 0 & 1/2 \\
        1/2 & 0 & -\sqrt{2}/2 & -1/2 \\
        1/2 & \sqrt{2}/2 & 0 & 1/2 \\
        1/2 & 0 & \sqrt{2}/2 & -1/2 \\
      \end{bmatrix},
$$
which yields
$$
\bar m(\bs w)=(1-\mu)w_1^2+\mu (w_2^2+w_3^2)+(1+\mu)w_4^2.
$$
From the condition $\bs w\cdot\bs U^\top\bs 1=1$ we find that $w_1=1/2$. Simplex $S_4$ is mapped to the convex set
$$
W_4=\{\bs w\colon -\frac 14-\frac{\sqrt{2}}{2}w_2+\frac 12 w_4\geq 0,\,\frac 14-\frac{\sqrt{2}}{2}w_3-\frac 12w_4\geq 0,\,\frac 14+\frac{\sqrt{2}}{2}w_2-\frac 12w_4\geq 0\}.
$$
If $|\mu|<1$ then the maximum value of the fitness surface is reached at $w_4=0$. The set $W_4$ is this case represents a square in the plane $(w_2,w_3)$ with the vertices at $p_i=\{\pm \sqrt{2}/4,\pm\sqrt{2}/4\}$, and at all fours vertices of this square the fitness surface attains local maxima
$$
\bar{m}(p_i)=\frac 14(1-\mu)+\frac 14 \mu=0.25.
$$
The internal equilibrium $\bs{\hat u}=(1/4,1/4,1/4,1/4)$ in the new variables is $\bs{\hat{w}}=(1/2,0,0,0)$, moreover
$$
\bar{m}(\bs{\hat{w}})=\frac{1}{4}-\frac{\mu}{4}\,,
$$
and we know that this equilibrium for $-1<\mu<0$ is asymptotically stable. In this case the fitness surface at this point has the global maximum, since $\bar{m}(\bs{\hat{w}})>\bar{m}(p_i)$. In general we have a system of five local summits, among which the central one has the highest fitness, and while the orbits $\gamma_t$ approach the equilibrium, the curves $\Gamma_t$ on $\Sigma$ can both increase and decrease while approaching $\bs{\hat w}$.

If $0<\mu<1$ then the curves $\Gamma_t$ will be approaching a closed curve that surrounds $\bs{\hat w}$, moreover, the maxima at the points $p_i,i=1,\ldots,4$ will be outside of this curve. On simplex $S_4$ the points $p_i$ correspond to the points $P_1=(0,0,1/2,1/2),P_2=(1/2,0,0,1/2),P_3=(1/2,1/2,0,0),P_4=(0,1/2,1/2,0)$. If $\mu>1$ then $P_i$ still give the maximal values of $\bar m(\bs u)$, however, $\bar m(\bs{\hat w})<0$ and now it is a global minimum of the fitness surface. This case corresponds to the heteroclinic cycle.

Finally, if $\mu<-1$ then $\bs{\hat u}$ is unstable and the maximum value of the form $\bar{m}(\bs w)$ is reached when $w_2=w_3=0$. It is reached at the points $Q_{1,2}=(1/2,0,0,\pm 1/2)$, and
$$
\bar{m}(Q_{1,2})=-\frac{\mu}{2}
$$
and hence $\bar m(Q_{1,2})>\bar m(\bs{\hat{w}})$. In this case we also have a heteroclinic cycle.

A schematic illustration of the system's dynamics and the corresponding structure of the mean fitness surface is given in Fig. \ref{fig0:4}.
\begin{figure}[!t]$\quad$
\includegraphics[width=0.45\textwidth]{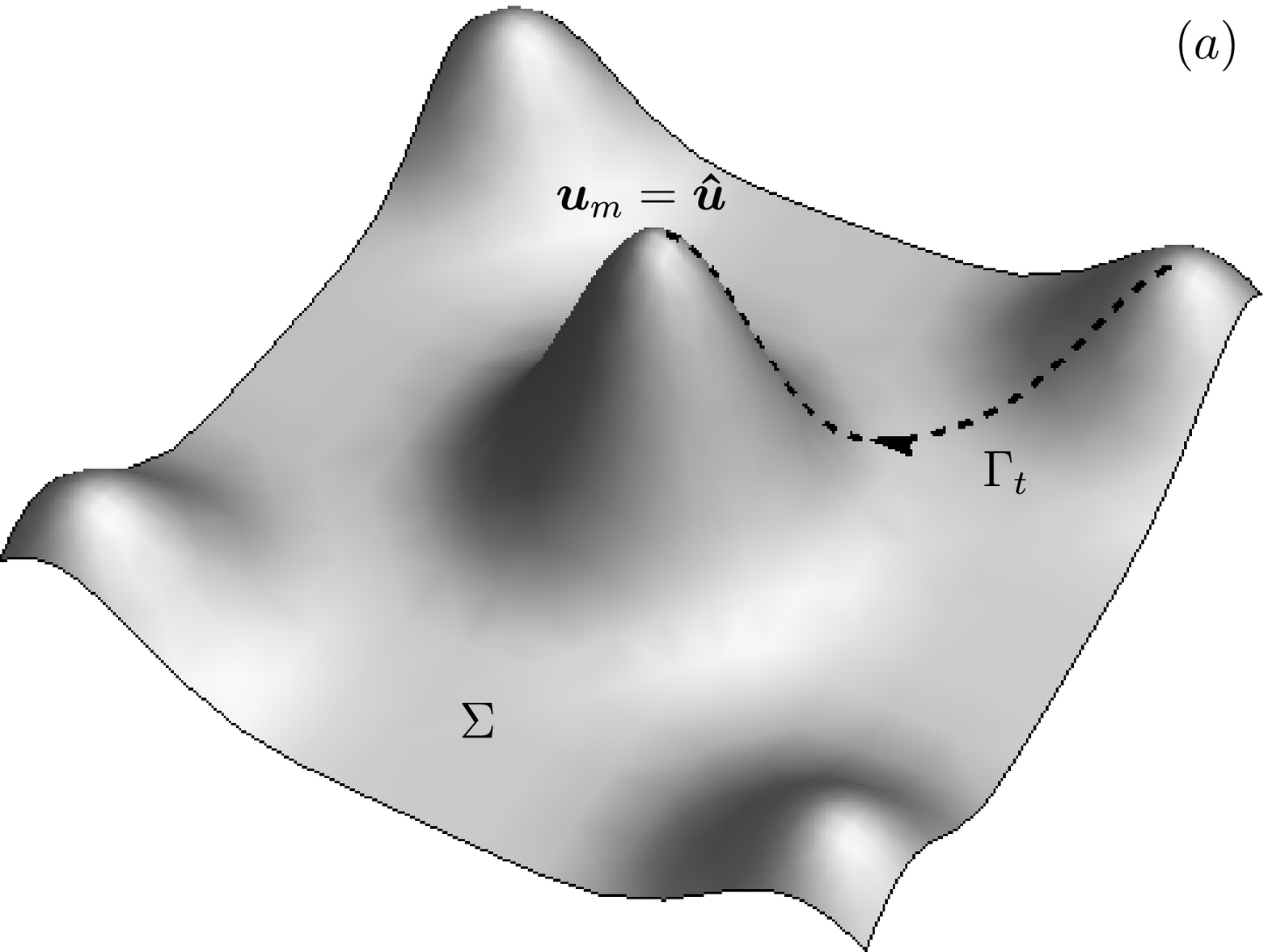}\hfill
\includegraphics[width=0.45\textwidth]{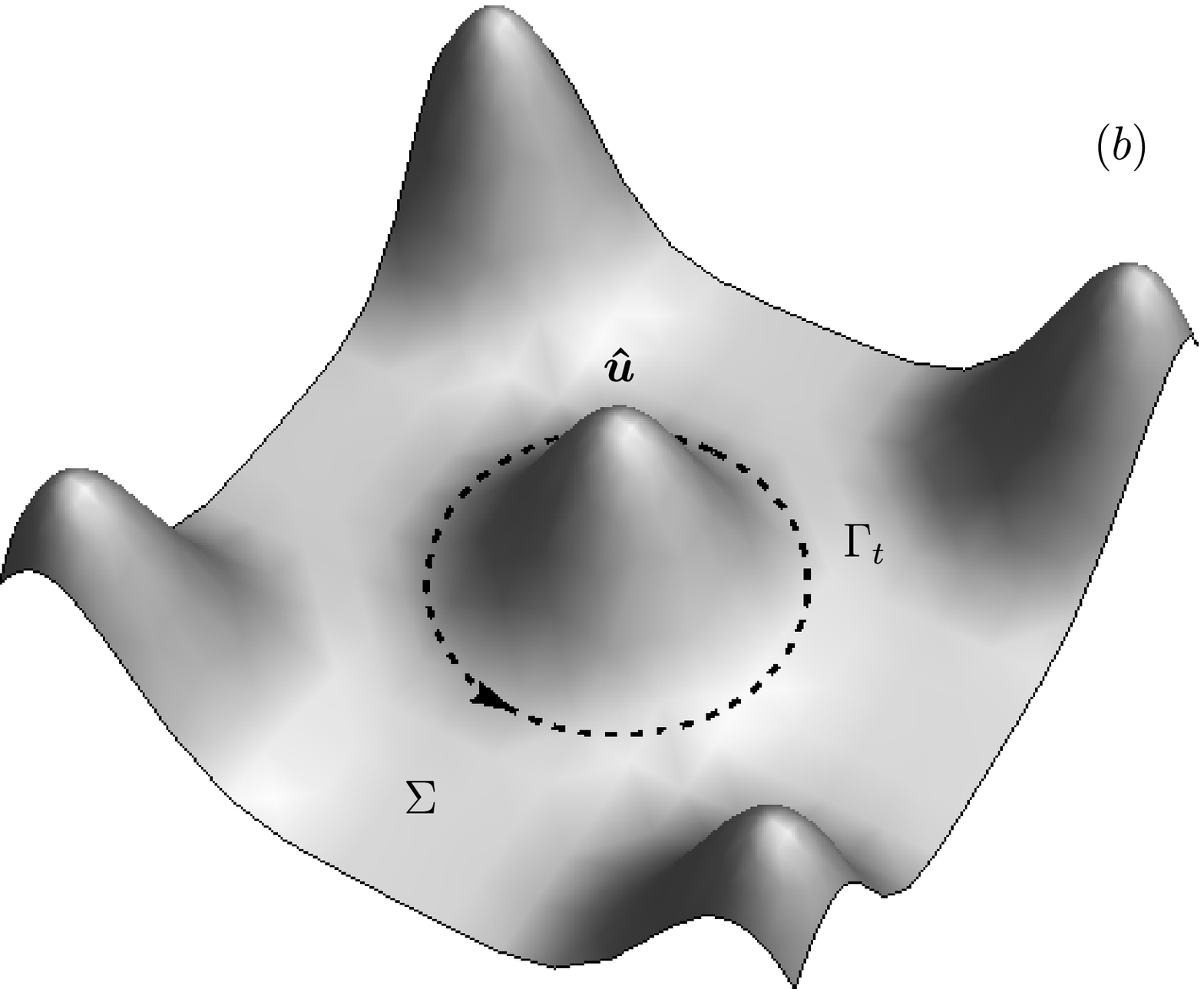}$\quad$
\caption{A schematic representation of the system's dynamics in Example \ref{ex3:2}. $(a)$ The unique internal equilibrium is also coincides with global fitness maximum and globally asymptotically stable. $(b)$ Now the internal equilibrium is no longer a global maximum, the solutions are now attracted to the limit cycle.}\label{fig0:4}
\end{figure}
\end{example}
%\bibliography{prebiotic}

\end{document}